\documentclass[conference]{IEEEtran}
\usepackage{cite}
\usepackage{amsmath,amssymb,amsfonts}
\usepackage{graphicx}
\usepackage{textcomp}
\usepackage{xcolor}
\def\BibTeX{{\rm B\kern-.05em{\sc i\kern-.025em b}\kern-.08em
    T\kern-.1667em\lower.7ex\hbox{E}\kern-.125emX}}

\usepackage{float}

\usepackage{hyphenat}
\usepackage{xspace}

\usepackage{algorithm}
\usepackage[noend]{algpseudocode}

\usepackage{amsthm}
\theoremstyle{plain}
\newcounter{thm}
\newtheorem{theorem}{Theorem}
\newtheorem{definition}[thm]{Definition}
\newtheorem{corollary}[thm]{Corollary}

\newcommand{\hide}[1]{}

\newcommand{\xhdr}[1]{\vspace{1.7mm}\noindent{{\bf #1.}}}

\newcommand{\eg}{{e.g.}\xspace}

\newcommand{\Secref}[1]{Sec.~\ref{#1}}

\newcommand{\Thmref}[1]{Thm.~\ref{#1}}

\newcommand{\Appref}[1]{Appendix~\ref{#1}}

\newcommand{\Algref}[1]{Alg.~\ref{#1}}
\newcommand{\Eqref}[1]{Eq.~\ref{#1}}

\usepackage{bbm} 
\usepackage{bm}
\usepackage{mathtools}
\DeclarePairedDelimiter\abs{\lvert}{\rvert}
\DeclarePairedDelimiter\norm{\lVert}{\rVert}

\newcommand{\defeq}{\mathrel{\vcentcolon=}}

\DeclareMathOperator{\Range}{Range}

\newcommand{\x}{\bm{x}}
\newcommand{\xp}{\bm{x'}}

\newcommand{\bb}{\bm{b}}
\newcommand{\cc}{\bm{c}}

\usepackage{cryptocode}
\usepackage{wrapfig}

\DeclareMathOperator{\compose}{Compose}
\DeclareMathOperator{\composegeneralized}{ComposeGeneralized}
\DeclareMathOperator{\refinetuples}{RefineTuples}

\DeclareMathOperator{\supp}{supp}
\DeclareMathOperator{\compguarantee}{CompGuarantee}

\newcommand{\rnum}[1]{\uppercase\expandafter{\romannumeral #1\relax}}

\begin{document}

\title{On the Choice of Databases in Differential Privacy Composition}

\author{\IEEEauthorblockN{Valentin Hartmann}
\IEEEauthorblockA{\textit{EPFL}\\
valentin.hartmann@epfl.ch}
\and
\IEEEauthorblockN{Vincent Bindschaedler}
\IEEEauthorblockA{\textit{University of Florida}\\
vbindsch@cise.ufl.edu}
\and
\IEEEauthorblockN{Robert West}
\IEEEauthorblockA{\textit{EPFL}\\
robert.west@epfl.ch}
}

\maketitle

\begin{abstract}
    Differential privacy (DP) is a widely applied paradigm for releasing data while maintaining user privacy. Its success is to a large part due to its composition property that guarantees privacy even in the case of multiple data releases.
Consequently, composition has received a lot of attention from the research community: there exist several composition theorems for adversaries with different amounts of flexibility in their choice of mechanisms. But apart from mechanisms, the adversary can also choose the databases on which these mechanisms are invoked. The classic tool for analyzing the composition of DP mechanisms, the so\hyp called composition experiment, neither allows for incorporating constraints on databases nor for different assumptions on the adversary's prior knowledge about database membership. We therefore propose a generalized composition experiment (GCE), which has this flexibility. We show that composition theorems that hold with respect to the classic composition experiment also hold with respect to the worst case of the GCE. This implies that existing composition theorems give a privacy guarantee for more cases than are explicitly covered by the classic composition experiment. Beyond these theoretical insights, we demonstrate two practical applications of the GCE: the first application is to give better privacy bounds in the presence of restrictions on the choice of databases; the second application is to reason about how the adversary's prior knowledge influences the privacy leakage. In this context, we show a connection between adversaries with an uninformative prior and subsampling, an important primitive in DP.
To the best of our knowledge, this paper is the first to analyze the interplay between the databases in DP composition, and thereby gives both a better understanding of composition and practical tools for obtaining better composition bounds.
\end{abstract}

\begin{IEEEkeywords}
differential privacy, differential privacy composition, databases
\end{IEEEkeywords}

\section{Introduction}
\label{sec:introduction}
Since its invention in 2006, differential privacy (DP) \cite{dwork2006calibrating} has become the de\hyp facto standard for releasing aggregate information about data in a privacy\hyp preserving way. It has been used in frequency estimation \cite{wang2017locally}, frequent itemset mining \cite{su2015differentially}, supervised \cite{abadi2016deep} and unsupervised \cite{su2016differentially} learning, and graph analysis \cite{kasiviswanathan2013analyzing}, to just name a few applications. Building blocks such as the Laplace mechanism \cite{dwork2006calibrating} or the Gauss mechanism \cite{dwork2014differential} allow for easily enhancing existing data analysis and machine learning (ML) methods with DP guarantees. DP has already been adopted by government agencies and companies such as the U.S. Census Bureau \cite{dajani2017modernization}, Google \cite{erlingsson2014rappor}, Microsoft \cite{ding2017collecting}, and Apple \cite{apple2017learning}.

The idea of DP is to hide the contribution of any single individual to a database: when computing queries on the database, noise is introduced into the process to make it impossible for an adversary with access to the query results to make high\hyp confidence statements about the contribution of any particular individual to the database. An algorithm that computes a noisy query result and fulfills DP is called a DP mechanism and parameterized by a tuple \((\varepsilon,\delta)\) that determines its level of privacy, where smaller parameters mean more privacy.
The original definition of DP only covers the release of a single mechanism output. In most practical scenarios, however, a database is queried more than once, and a privacy guarantee needs to be given over the composition of all DP mechanisms that are invoked on it. Note that this can even be the case when in the end only a single datum is released, because the data analysis process itself might require multiple passes over the data, e.g., when training an ML model via gradient descent. The composition of DP mechanisms is usually analyzed using the so\hyp called \emph{composition experiment} \cite{dwork2010boosting}. The \((\varepsilon,\delta)\)\hyp guarantee that can be given for the composition of \(k\) mechanisms with individual guarantees \((\varepsilon_1,\delta_1),\dots,(\varepsilon_k,\delta_k)\) not only depends on the magnitude of the \((\varepsilon_i,\delta_i)\), but also on whether the mechanisms are all the same or may differ, whether the \((\varepsilon_i,\delta_i)\) are all the same or may differ, and whether the \((\varepsilon_i,\delta_i)\) are fixed beforehand or not. Based on the restrictions on the mechanisms' DP guarantees, different composition theorems can be applied \cite{sommer2019privacy,kairouz2017composition,complexity2016,rogers2016privacy}. Tighter restrictions typically lead to better privacy guarantees.

Besides invoking \emph{mechanisms} with different properties in different iterations, invoking those mechanisms on different \emph{databases} or on different subsets of the same database has also been investigated. Examples are the sparse vector technique \cite{dwork2009complexity,svt2017}, the composition of top\hyp \(k\) queries \cite{durfee2019practical}, the sample and aggregate framework \cite{nissim2007smooth} or the sampling of mini\hyp batches in differentially private SGD \cite{abadi2016deep,mironov2019renyi}.
Analogously to restricting the set that the mechanisms can be selected from, one can also think of restricting the databases on which these mechanisms are invoked. However, while restrictions on the mechanisms have received a lot of attention from the research community, restrictions on the databases have not. In this paper, we define a generalization of the composition experiment regarding the databases on which the mechanisms are invoked. This allows us to analyze the influence of the database choices and of an adversary's prior knowledge about contributions of individuals, on the DP guarantee. This is necessary because the classic composition theorem assumes a particular relationship of the databases to each other and a particular type of adversary, and does not cover certain relevant settings, for which it gives overly loose privacy bounds.

\xhdr{Example}
Consider the following example: a government agency wants to compare the quality of the stationary care of the different hospitals in the country. For this, each hospital collects data about their stationary patients (i.e., patients that stay for at least one night) over the span of a year. If there are \(k\) hospitals in the country, there are hence \(k\) separate databases. For simplicity, assume that a single DP mechanism \(M\) that returns a score for the quality of care is invoked on each of the databases. For obtaining a privacy guarantee via the classic composition experiment, \(k\) invocations of \(M\) would have to be composed. However, there exists a natural constraint on the data contributed by each individual: there are only \(365\) nights in a year, so one individual can be a stationary patient in at most \(365\) hospitals over the course of a year. Intuitively, if \(k>365\), only \(365\) invocations of \(M\) would have to be composed. In this paper, we formalize this intuition via a generalized composition experiment, which yields tighter privacy bounds in such cases than the classic composition experiment.\footnote{Note that in this example one could apply the parallel composition theorem \cite{mcsherry2009privacy}, which also allows for reducing the number of compositions to 365. However, this theorem only allows for simple composition---besides other downsides---, whereas our generalized composition experiment does not have this restriction and is compatible with advanced composition. See \Secref{sec:parallel} for an in\hyp depth comparison with parallel composition.}

\subsection{Summary of Contributions}

Our key contributions are as follows.

\begin{itemize}
    \item After giving some background on DP composition and related work in \Secref{sec:background_related}, we analyze the semantics of the classic composition experiment and point out its shortcomings due to not considering database choices in \Secref{sec:limitations}.
    \item In \Secref{sec:generalized_comp} we then define a \emph{generalized composition experiment} (GCE), which comes without these shortcomings and allows for incorporating the database choices into the DP analysis.
    \item Next, we show that composition theorems that hold with respect to the classic composition experiment also hold with respect to the worst case of the GCE (\Secref{sec:relationship}).
    \item We analyze two application of the GCE (\Secref{sec:applications}):
    \begin{itemize}
        \item We can give better privacy guarantees when there are constraints on an individual's contribution to the different databases (\eg, an individual cannot have contributed to all databases at once; \Secref{sec:constraints}).
        \item We get a better understanding of the knowledge gain of adversaries with different amounts of prior knowledge (\Secref{sec:adversaries}). In this context, we show a connection between uncertainty in the prior knowledge of the adversary and uncertainty through subsampling of the data.
    \end{itemize}
\end{itemize}

As such, the generalized composition theorem introduced in this paper is both a tool to better understand DP composition and a tool to obtain better composition bounds.

\section{Background \& Related Work}
\label{sec:background_related}
\subsection{Background}
Pure DP was first defined in 2006 by Dwork et al. \cite{dwork2006calibrating} with only a single parameter \(\varepsilon\), and later relaxed to approximate DP, which allows for some slack \(\delta\) \cite{dwork2006our}. Since setting \(\delta\) to \(0\) turns approximate DP into pure DP, we will work with the \((\varepsilon,\delta)\) definition and refer to it simply as DP. Whenever we explicitly mean \(\delta=0\), we will write \(\varepsilon\)\hyp DP instead of \((\varepsilon,\delta)\)\hyp DP.

DP aims at hiding the contribution of each single individual to a database by ensuring that each output of a randomized algorithm that is executed on the database is similarly likely no matter whether the individual contributed to the database or not. This is formalized by requiring the output distributions resulting from executing the algorithm on so\hyp called neighboring databases to be close. If an individual can influence at most one record in the database, then natural choices for the neighborhood relationship are either to define databases as neighboring if they have the same size but differ in at most one record (bounded DP) or if one database can be obtained from the other by adding at most a single record (unbounded DP). For a simpler exposition, we will restrict ourselves to unbounded DP throughout this paper unless stated otherwise. DP mechanisms are parameterized by the two parameters \(\varepsilon\) and \(\delta\). Typical choices for \(\varepsilon\) and \(\delta\) are \(\varepsilon\leq 1\) and \(\delta\ll 1/|\x|\), where smaller values mean more privacy ($\x$ is the input database).

\begin{definition}[Differential privacy {\cite{dwork2006our}}]
\label{def:dp}
A randomized algorithm \(M\) with domain \(\mathcal{D}\) is \((\varepsilon,\delta)\)\hyp {\em differentially private} if, for all \(S\subset \Range(M)\) and for all neighboring databases \(\x,\x'\in\mathcal{D}\),
\begin{equation*}
    \Pr(M(\x)\in S)\leq e^{\varepsilon}\Pr(M(\x')\in S) + \delta.
\end{equation*}
\end{definition}

A randomized algorithm that fulfills DP is called a DP mechanism. In practice, one usually wants to not only release a single piece of information about a database, but multiple pieces, e.g., multiple summary statistics of multiple columns. Also, more complex algorithms such as the training of an ML model via gradient descent might require multiple accesses to the database. Thus, for practical purposes, it is paramount to give privacy guarantees that hold over multiple mechanism outputs. One way to extend DP in this direction is to define a sequence of DP mechanisms \(M_1,\dots,M_k\) with domains \(\mathcal{D}_1,\dots,\mathcal{D}_k\) to be differentially private if the mechanism \((M_1,\dots,M_k)\) with domain \(\mathcal{D}_1\times\dots\times\mathcal{D}_k\) and range \(\Range(M_1)\times\dots\times\Range(M_k)\) is differentially private. However, in this form of composition everything has to be fixed beforehand: neither can one choose \(M_i\) based on the outputs of \(M_1,\dots,M_{i-1}\), nor the database on which \(M_i\) is invoked.

Instead, we would like to allow for the following, more flexible way of accessing data:
There are two parties, a data analyst and a data curator. The data curator has access to private databases, while the data analyst does not have access to any private data. However, the data analyst knows which databases the data curator has access to. The data analyst gets access to private data via the data curator \(k\) times. In round \(i\) of \(k\), the data analyst chooses one of the private databases and one mechanism \(M_i\) from a class \(\mathcal{M}_i\) of differentially private mechanisms. The data curator invokes \(M_i\) on the chosen database and returns the result to the data analyst. The goal now is to bound the amount of information that the data analyst learns from the mechanism outputs about the contributions of single individuals to the databases.

\begin{algorithm}
\caption{\(\compose(\mathcal{A},(\mathcal{M}_1,\dots\mathcal{M}_k),k,b)\)}
\label{alg:comp_exp}
\begin{algorithmic}[1]
\State{\textbf{Input:} \(\mathcal{A},(\mathcal{M}_1,\dots\mathcal{M}_k),k,b\in\{0,1\}\)}
\State{\textbf{Output:} \(V^b\)}
\State{Select coin tosses \(r\) for \(\mathcal{A}\) uniformly at random}
\For{\(i=1,\dots,k\)}
    \State{\(\mathcal{A}\) returns neighboring databases \(\x_i^0,\x_i^1\), parameters \(w_i\), and a mechanism \(M_i\in\mathcal{M}_i\)}
    \State{\(\mathcal{A}\) receives \(y^b_i \sim M_i(\x_i^b,w_i)\)}
\EndFor
\Return{View \(v^b = (r, y^b_1, \dots, y^b_k)\)}
\end{algorithmic}
\end{algorithm}

To analytically reason about this setting, Dwork et al. \cite{dwork2010boosting} introduced the so\hyp called \emph{composition experiment} (\Algref{alg:comp_exp}). It is a hypothetical game and not executed like this in practice; it is merely a tool to analyze the privacy leakage over repeated accesses to private databases.
In the composition experiment, the data analyst is the adversary \(\mathcal{A}\). The adversary is allowed to, in each step $i$, pick any pair of neighboring databases \(\x_i^0,\ \x_i^1\) and a mechanism \(M_i\) out of a class of mechanisms \(\mathcal{M}_i\), together with parameters \(w_i\). This can be done adaptively based on the previous mechanism outputs. Typically, the adversary would choose a database that contains the record of a target individual that the adversary is interested in, and the neighboring database that does not contain this record. The adversary communicates the databases, the mechanism and the parameters to the data curator, who returns  \(y^b_i \sim M_i(\x_i^b,w_i)\) based on a private bit \(b\) that is only known to the data curator. The adversary wants to learn the value of \(b\) from \(y^b_1, \dots, y^b_k\), and DP for composition (Def.~\ref{def:composition}) ensures that the adversary cannot learn the value of \(b\) with high confidence. Why does this suffice to guarantee privacy in the setting described in the previous paragraph? Assume that the adversary suspects that their target contributed to all \(k\) databases. Then the adversary could choose \(\x_1^0,\dots,\x_k^0\) to be the databases that do not contain the record of the target, and \(\x_1^1,\dots,\x_k^1\) to be the same databases but with the record added. In the real world, this corresponds to an adversary that can actively influence all records in the databases, except that they do not know whether or not there is one additional record that was added by their target. If a guarantee against such a strong adversary holds, then it also holds against weaker adversaries that, for example, cannot influence the database or even only know parts of it. Note that the adversary may also choose \(\x_i^0\) to be the database that contains the target's record and \(\x_i^1\) to be the database that does not contain it, if they suspect that the target did not contribute to the \(i\)\hyp th database.
In Def.~\ref{def:composition}, the definition of DP for the composition of mechanisms, we summarize the randomness \(r\) of the adversary and the mechanism outputs in the view \(v^b = (r, y^b_1, \dots, y^b_k)\). \(r\) is required for the adversary to reconstruct their choices throughout the experiment. We denote the corresponding random variables with capital letters, i.e., \(V^b = (R, Y^b_1, \dots, Y^b_k)\).
\begin{definition}
\label{def:composition}
A sequence \(M_1,\dots,M_k\) of mechanisms is \((\varepsilon,\delta)\)\hyp {\em differentially private} if, for all sets \(S\) of views,
\begin{equation*}
    \Pr(V^0\in S) \leq e^{\varepsilon} \Pr(V^1\in S) + \delta.
\end{equation*}
and
\begin{equation*}
    \Pr(V^1\in S) \leq e^{\varepsilon} \Pr(V^0\in S) + \delta.
\end{equation*}
\end{definition}

The earliest result on composition is the simple composition theorem \cite{dwork2006our,dwork2010boosting}, which states that one can add up the \(\varepsilon\)'s and \(\delta\)'s of individual mechanisms to get a DP guarantee for their composition. For the composition of \(k\) \((\varepsilon,\delta)\)\hyp differentially private mechanisms, simple composition would result in a bound of \((k\varepsilon,k\delta)\). Later, Dwork et al.\ \cite{dwork2010boosting} showed an asymptotically better bound, the advanced composition theorem, which gives a guarantee of \((\varepsilon',k\delta + \delta')\), for any \(\delta'>0\) and for
    $\varepsilon' = \sqrt{2k\ln(1/\delta')}\varepsilon + k\varepsilon(e^{\varepsilon}-1)$
in the same setting.

\subsection{Related work}
\label{sec:relwork}

\xhdr{Composition theorems}
While better than simple composition, the advanced composition theorem is not optimal.
There are different levels of flexibility that one can give the adversary, which lead to different composition bounds. In the homogeneous setting, all mechanisms \(M_i\) need to have the same DP\hyp guarantee, i.e., \(\mathcal{M}_i = \mathcal{M}_j\) for all \(i,j\) and \(\mathcal{M}_i\) contains all \((\varepsilon,\delta)\)\hyp DP mechanisms for a fixed pair \((\varepsilon,\delta)\).
The optimal composition theorem for homogeneous composition was proved by Kairouz et al. \cite{kairouz2017composition} by modeling differential privacy in terms of hypothesis testing. A different proof for the same theorem that uses more standard DP arguments was later given by Murtagh et al. \cite{complexity2016}.

In the heterogeneous case, each \(\mathcal{M}_i\) consists of all \((\varepsilon_i,\delta_i)\)\hyp DP mechanisms, where \((\varepsilon_i,\delta_i)\) is fixed ahead of time, but may be different for different \(i\).
Murtagh et al.\ \cite{complexity2016} proved the optimal composition theorem for the heterogeneous case. Computing the DP\hyp guarantee is \(\#P\)\hyp complete and thus requires time exponential in \(k\) to compute, but the authors give a polynomial approximation algorithm.


A composition theorem with a different flavor that, as our paper, is concerned with the choice of databases, is parallel composition~\cite{mcsherry2009privacy}. It states that if in all steps an \(\varepsilon\)\hyp differentially private mechanism is invoked on disjoint subsets of the same database that result from splitting the data domain into disjoint subsets, then the composition fulfills \(\varepsilon\)\hyp DP. That is, one only needs to account for a single mechanism invocation. We discuss parallel composition and how it relates to the GCE in detail in \Secref{sec:parallel}.

\xhdr{Assumptions on adversaries}
One of the things that our GCE allows to analyze are the privacy implications from modifying the adversary's prior knowledge or goals (see \Secref{sec:adversaries}).
Prior work in that direction includes membership privacy \cite{li2013membership}, which allows for modeling, e.g., settings where only the contribution or the non\hyp contribution of an individual to a database should be kept secret from the adversary, whereas DP protects both.
Noiseless DP \cite{bhaskar2011noiseless,duan2009privacy} relaxes the assumption that the adversary can influence all records in the database except from one, and instead assumes that the adversary only knows the distribution from which the database is drawn.
Desfontaines et al. \cite{desfontaines2019differential} consider adversaries with only partial knowledge about the database, and distinguish between an active adversary that can influence the database, and a passive adversary that cannot.
The very general pufferfish framework \cite{kifer2014pufferfish} that subsumes many privacy definitions allows, like our framework, for specifying constraints on the database and on the adversary's beliefs. However, it does this for the case of mechanism invocations on a single database. The authors briefly discuss the case of invoking mechanisms on different databases and the importance of taking the relationships between those databases into account, but without providing specialized composition theorems for this case.

\xhdr{Subsampling}
In \Secref{sec:adversaries} we show a connection between subsampling of database records and assumptions about the adversary's prior knowledge. Executing a mechanism only on a subsample of the database instead of the entire database improves the DP guarantee. This has been used, e.g., for stochastic gradient descent \cite{abadi2016deep}. Balle et al. \cite{balle2018privacy} showed tight bounds for subsampling. Subsampling has also been analyzed for variants of DP that are focused on composition, e.g., for R\'enyi DP \cite{mironov2019r} or for truncated concentrated DP \cite{bun2018composable}.

\section{Limitations of the Classic Composition Experiment}
\label{sec:limitations}
In this section, we give insights into what the classic composition experiment models and what it does not model. This is accompanied by examples that show how this translates to real\hyp world settings. In \Secref{sec:generalized_comp}, we then define the more flexible GCE that can additionally capture settings for which the classic composition experiment is too restrictive.

What both the homogeneous and the heterogeneous composition setting (see \Secref{sec:relwork}) have in common is that they rely on the---hypothetical---composition experiment (\Algref{alg:comp_exp}) and the corresponding definition of DP (Def.~\ref{def:composition}). In particular, at each step they let the adversary pick a pair of neighboring databases \(\x_i^0,\x_i^1\). These are the two databases between which the adversary wants to differentiate. It is assumed that there is a single bit \(b\) that determines the ``true'' database on which the mechanism is invoked, and it determines this for all iterations at the same time. This means that, depending on \(b\), either all mechanism are invoked on the first database in their iteration or all are invoked on the second database. A bound---and from it, a DP guarantee---is then derived for the ability of the adversary to tell what the value of \(b\) is based on the outputs of the mechanisms.

Let us take a step back and think about how this translates to the real world. Remember that for us neighboring databases are databases that are identical except for one record that is present in one database but not in the other. The private information to be protected is whether or not a particular record is part of the database. Assume that \((\x_i^0,\x_i^1) = (\x^0,\x^1)\) for fixed \((\x^0,\x^1)\) for all \(i\), i.e., the adversary always picks the same pair of databases, and thus the record that is present in one and missing in the other is also always the same. Then Def.~\ref{def:composition} makes perfect sense: a sequence of mechanisms that is differentially private according to that definition limits the ability of the adversary to guess whether the record in which the databases differ is part of the database \(\x^b\) on which the data curator invokes the mechanisms or not.

Consider now the setting where the database pairs are not all the same, and for simplicity assume that no two of the databases chosen by the adversary are the same, i.e., \(\x_i^0 \neq \x_j^0\), \(\x_i^0 \neq \x_j^1\) and \(\x_i^1 \neq \x_j^1\) for all \(i\neq j\).
In each iteration \(i\), the adversary targets one record \(x_i\) about whose database membership they want to learn something. As the pair of databases in iteration \(i\), they choose a database that does not contain the record and the same database with this record added. But in addition to the two databases, the adversary also has to commit to an order of this database pair. We can define corresponding hypotheses \(H_{0,i}, H_{1,i}\) for each iteration \(i\) that the adversary commits to, as follows: If \(\x_i^0\) is the database that does not contain record \(x_i\), then \(H_{0,i}\) is ``the database on which the \(i\)\hyp th mechanism is invoked does not contain \(x_i\)'' and \(H_{1,i}\) is ``the database on which the \(i\)\hyp th mechanism is invoked contains \(x_i\)''. If \(\x_i^0\) is the database that contains \(x_i\), then we swap the two hypotheses. Let \(H_0 = H_{0,1} \land \dots \land H_{0,k}\) be the hypothesis that all of \(H_{0,1},\dots,H_{0,k}\) are true, and \(H_1 = H_{1,1} \land \dots \land H_{1,k}\) be the hypothesis that all of \(H_{1,1},\dots,H_{1,k}\) are true. Due to the way the composition theorem is set up (there is only a single bit \(b\)), either \(H_0\) is correct or \(H_1\) is correct. If the sequence of mechanisms fulfills Def.~\ref{def:composition}, we have a bound on the adversary's ability to determine which of \(H_0\) and \(H_1\) is correct. This raises the questions: In practice, is it realistic to assume that the adversary can reduce the set of possible hypotheses to just these two? And is it even always the case that either \(H_0\) or \(H_1\) is true? Let \(B=\{0,1\}^k\) be the set of all vectors of length \(k\) with binary entries. For any \(\bb\in B\), let
\begin{equation*}
    H_{\bb} = \bigwedge\limits_{i=1}^k H_{b_i,i}.
\end{equation*}
Could it not be the case that \(H_{\bb}\) for \(\bb\notin\{\vec{0},\vec{1}\}\) is true (\(\vec{0}\) and \(\vec{1}\) denote the 0- and the 1\hyp vector here, respectively)?

\subsection{Examples}
\xhdr{Example 1}
There are some cases where it makes sense to assume that only \(H_0\) or \(H_1\) can be true, even if the queried database is not the same in each iteration. Assume that the individual about whose contribution to databases the adversary wants to learn something is the same in each iteration (though the specific record contributed by this individual may be different in different databases). Say the databases \(\x_i^0,\ \x_i^1\) contain usage statistics of individual users of software applications. If the databases all correspond to one of the different applications contained in an office suite that is sold as a package, then it is reasonable to assume that if a user's data is part of one of the databases, then it is part of all databases, because customers only have the option to either buy all applications bundled in the office suite or none of them. However, if the databases correspond to software applications that are sold separately, then things are different. Without additional prior information, it is no longer reasonable to assume that the targeted individual either uses/owns all applications or none of them. It might instead be much more likely that the individual uses either none or a strict subset of them, but not all. What can, e.g., be said for sure is that the individual either uses no application or at least one application. And knowing whether an individual uses at least one application could indeed be valuable information for an adversary. The adversary might, e.g., know about a vulnerability in a code library used by all applications of this software vendor, and wants to know whether the individual could be attacked this way, because they use one of the vulnerable applications. Another example are police databases from different states and an adversary that wants to know whether the target individual committed a crime in at least one state. Assume that the adversary in each iteration \(i\) picks \(\x^0_i\) to be the database not containing the individual's data, and \(\x^1_i\) to be the database containing the individual's data, and that the adversary only has an uninformative prior about database membership of the target. Then this would correspond to testing hypothesis \(H_0\) vs.\ \(U(\{H_{\bb} \mid \bb\in B\setminus\{\vec{0}\}\})\), i.e., the uniform distribution over all hypotheses except \(H_0\). But this cannot be captured by a single run of \Algref{alg:comp_exp}, because \Algref{alg:comp_exp} only allows for testing non\hyp composite hypotheses: in Def.~\ref{def:composition}, the hypotheses that are compared are \(b=0\) versus \(b=1\).

\xhdr{Example 2}
There are also settings where it is of interest whether the individual contributed to at least a certain number \(k'\) of the databases. If the databases contain the data of one medical practice each, an employer could try to find out whether a job applicant has gone to a large number of different medical doctors in the past, which might indicate severe health problems and therefore many sick days. This would correspond to hypothesis \(U(\{H_{\bb} \mid \norm{\bb}_1<k'\})\) vs.\ \(U(\{H_{\bb} \mid \norm{\bb}_1\geq k'\})\). We hence need a more general composition guarantee that can capture more hypotheses than just \(H_0\) vs.\ \(H_1\).

To give a different angle at the problem: A main idea behind the development of differential privacy was to give a privacy definition that assumes the least possible about the adversary's knowledge and about what the adversary wants to learn. DP guarantees even hold when the adversary knows all records that are in the database except one, and no matter what the adversary wants to learn about the target record. But when it comes to composition, the current model via the composition experiment makes the assumption that the adversary only wants to compare the hypotheses with associated \(\bb\)\hyp vectors \(\vec{0}\) and \(\vec{1}\). This leads us to the GCE in \Secref{sec:generalized_comp}, which does not have this restriction. As we will prove later, the assumption of the classic composition theorem that the adversary only wants to test non\hyp composite hypotheses does not undermine privacy. This is one of the contributions of this paper. Apart from that, our GCE has multiple applications --- e.g., for proving better privacy bounds in certain situations ---, two of which we analyze in \Secref{sec:applications}.

\xhdr{Example 3}
We can use the GCE to prove a better privacy bound for the example with the hospitals from the introduction. We can model the constraint that the target could have contributed to at most \(365\) databases by restricting the set of possible hypotheses to those that are composed of at most \(365\) \(H_{1,i}\) hypotheses and otherwise \(H_{0,i}\) hypotheses, i.e., those \(H_{\bb}\) with \(\norm{\bb}_1\leq 365\).

\section{A Generalized Composition Experiment}
\label{sec:generalized_comp}
In order to capture the entire range of possible hypotheses, we need to generalize the composition experiment, which we do in \Algref{alg:comp_exp_gen}. Instead of a single bit \(b\), the input now contains a vector of bits \(\bb\), where the \(i\)\hyp th entry \(b_i\) determines which of the two databases the mechanism in iteration \(i\) is invoked on. The adversary tries to guess \(\bb\) from the mechanism outputs. The original composition experiment would correspond to restricting \(\bb\) to be the 0- or the 1\hyp vector. Hypotheses are given as distributions over \(\bb\) vectors, assigning to each vector a belief in the form of a probability. To simplify notation throughout the paper, we will often assume that database 0 has at most as many records as database 1: \(|\x_i^0|\leq |\x_i^1|\). We can do this without loss of generality, for if this inequality does not hold, we can simply flip the corresponding bit \(b_i\).

\begin{algorithm}[H]
\caption{\(\composegeneralized(\mathcal{A},(\mathcal{M}_1,\dots\mathcal{M}_k),k,\bb)\)}
\label{alg:comp_exp_gen}
\begin{algorithmic}[1]
\State{\textbf{Input:} \(\mathcal{A},(\mathcal{M}_1,\dots\mathcal{M}_k),k,\bb\in B=\{0,1\}^k\)}
\State{\textbf{Output:} \(V^{\bb}\)}
\State{Select coin tosses \(r\) for \(\mathcal{A}\) uniformly at random}
\For{\(i=1,\dots,k\)}
    \State{\(\mathcal{A}\) returns neighboring databases \(\x_i^0,\x_i^1\), parameters \(w_i\) and a mechanism \(M_i\in\mathcal{M}_i\)}
    \State{\(\mathcal{A}\) receives \(y^{b_i}_i \sim M_i(\x_i^{b_i},w_i)\)}
\EndFor
\Return{View \(v^{\bb} = (r, y^{b_1}_1, \dots, y^{b_k}_k)\)}
\end{algorithmic}
\end{algorithm}

In Def.~\ref{def:hdp}, we define privacy with respect to the GCE (\Algref{alg:comp_exp_gen}), in a way that generalizes the classic composition experiment (\Algref{alg:comp_exp}) and the corresponding privacy definition (Def.~\ref{def:composition}). Def.~\ref{def:hdp} bounds how much the likelihood of the mechanism outputs given one hypothesis \(p_0\) regarding \(\bb\), and the likelihood given any other hypothesis \(p_1\) regarding \(\bb\) may differ.

\begin{definition}[Hypothesis differential privacy]
\label{def:hdp}
A sequence of mechanisms in \Algref{alg:comp_exp_gen} is \emph{\((\varepsilon,\delta)\)\hyp hypothesis differentially private} with respect to a set \(\mathcal{P}\) of pairs of distributions over \(B\) if for all pairs of distributions \((p_0,p_1)\in\mathcal{P}\) and all sets of views \(S\),
\begin{equation}
    \mathbb{E}_{\bb\sim p_0}[\Pr(V^{\bb}\in S)] \leq e^{\varepsilon} \mathbb{E}_{\bb\sim p_1}[\Pr(V^{\bb}\in S)] + \delta\label{eq:hdp0}
\end{equation}
and
\begin{equation}
    \mathbb{E}_{\bb\sim p_1}[\Pr(V^{\bb}\in S)] \leq e^{\varepsilon} \mathbb{E}_{\bb\sim p_0}[\Pr(V^{\bb}\in S)] + \delta.\label{eq:hdp1}
\end{equation}
We say the sequence is \((\varepsilon,\delta)\)\hyp hypothesis differentially private with respect to a pair of distributions \((p_0,p_1)\) if it is \((\varepsilon,\delta)\)\hyp hypothesis differentially private with respect to the one\hyp element set \(\{(p_0,p_1)\}\).
\end{definition}

This is a guarantee against an adversary that has two (potentially composite) hypotheses about the databases in the different iterations, and wants to learn from the mechanism outputs which one is closer to the true \(\bb\). The definition allows for restricting the set of hypotheses that the adversary may have via the set \(\mathcal{P}\). This restriction can, e.g., result from:
\begin{enumerate}
    \item \textbf{External constraints on the databases.} In the example from the introduction, an individual can contribute to at most \(365\) of the databases. This implies---assuming \(|\x_i^0|\leq |\x_i^1|\)---\(p_0(\bb) = p_1(\bb) = 0\) for all \(\bb\) with \(\norm{\bb}_1>365\) for all \((p_0,p_1)\in \mathcal{P}\). In \Secref{sec:constraints} we analyze this and other examples in more detail.
    \item \textbf{Assumptions about the adversary.} There is a line of work analyzing how different assumptions about the adversary affect the privacy guarantee (see \Secref{sec:relwork}). With the generalized composition theorem one is able to analyze how different adversarial hypotheses about database membership affect the privacy guarantee, which we do in \Secref{sec:adversaries}.
\end{enumerate}

The classic composition experiment (\Algref{alg:comp_exp}) and privacy definition (Def.~\ref{def:composition}) are a special case of our generalized experiment and definition. They can be recovered by setting \(\mathcal{P}\) to be the set that contains only the pair of distributions \([p_0(\vec{0}) = 1,\ p_0(\bb) = 0\text{ for all }\bb\neq \vec{0}]\) and \([p_1(\vec{1}) = 1,\ p_1(\bb) = 0\text{ for all }\bb\neq \vec{1}]\).

Note that Def.~\ref{def:hdp} can be equivalently formulated as bounding the Bayes factor \cite{goodman1999toward} between any pair of hypotheses from \(\mathcal{P}\) by \(e^{\varepsilon}\), with an additional additive constant \(\delta\).

\subsection{Relating Classic and Generalized Composition}
\label{sec:relationship}
We show that any sequence of mechanisms fulfills \((\varepsilon,\delta)\)\hyp DP (Def.~\ref{def:composition}) if and only if it fulfills generalized \((\varepsilon,\delta)\)\hyp HDP (Def.~\ref{def:hdp}) with respect to the set of all possible pairs of hypotheses. Here we give a sketch of the proof; the full proof can be found in \Appref{app:proof_equivalence}.

\begin{theorem}\label{thm:equivalence}
Let \(\mathcal{D}(B)\) denote the set of all probability distributions over the set \(B\). A sequence of mechanisms is \((\varepsilon,\delta)\)\hyp differentially private if and only if it is \((\varepsilon,\delta)\)\hyp hypothesis differentially private with respect to the set \(\mathcal{P} = \mathcal{D}(B)\times \mathcal{D}(B)\).
\end{theorem}

\begin{proof}[Proof sketch]
It is easy to see that DP is a special case of HDP, where the two hypotheses are given by the 0- and the 1\hyp vector, which shows one direction of the theorem. For showing the other direction, we first only consider non\hyp composite hypotheses and progressively flip bits in the corresponding \(\bb\)\hyp vectors until we have two vectors to which we can apply the classic composition experiment. We then extend this analysis to tuples \((p_0,p_1)\) of composite hypotheses: \(p_0\) and \(p_1\) can be expressed in terms of the sets of tuples \(T_0 = \{(\bb_0, p_0(\bb_0))\mid \bb_0\in\supp(p_0)\}\) and \(T_1 = \{(\bb_1, p_1(\bb_1))\mid \bb_1\in\supp(p_1)\}\). In order to be able to apply the previously obtained inequalities, we need to be able to match tuples \((\bb_0,w_0)\in T_0\) and \((\bb_1,w_1)\in T_1\), where \(w_0=w_1\). However, this will not always be possible, since the distributions over probability masses need not be the same for \(p_0\) and \(p_1\). We can solve this problem by refining tuples via \Algref{alg:refine} from \Appref{app:refinement}. It takes as input the sets \(T_0\) and \(T_1\), and in each step refines a tuple \((\bb,w)\) from one of the sets into two tuples \((\bb,w^1),\ (\bb,w^2)\) with \(w = w^1+w^2\). It returns two sets \(U_0,\ U_1\) of such refined tuples for which a matching can be found.
\end{proof}

Let us summarize what we have done so far. We have explained why the classic privacy definition for the composition of DP mechanisms does not (explicitly) cover many cases of practical relevance. We have then generalized the definition to a privacy definition that covers these cases. Finally, we have shown that fortunately sequences of mechanisms that fulfill the original definition also fulfill the generalized definition in all cases. Hence, sequences of mechanisms that fulfill DP protect against a wider range of adversaries than what is explicitly modeled in the composition experiment. We now proceed by using our generalized definition to derive better privacy bounds in certain settings.

\section{Applications}
\label{sec:applications}
We analyze two applications of the generalized composition experiment.

\subsection{Better Bounds from Database Membership Constraints}
\label{sec:constraints}
\xhdr{Motivation}
In the example from the introduction, there were natural constraints on the contribution of an individual to the database: there are only \(365\) nights in a year and thus an individual can contribute to at most \(365\) hospital databases. Such restrictions on database contribution consequently restrict the set of pairs of hypotheses \(\mathcal{P}\) with respect to which an HDP guarantee needs to be given, which can result in better privacy bounds.
In this section, we analyze three examples.
We first need some notation.
\begin{definition}
A composition theorem is {\em compatible} with a sequence \(M_1,\dots,M_k\) of DP mechanisms if it is applicable to the sequence \(M_1,\dots,M_k\) in the context of the classic composition experiment.
\end{definition}

For example, if \(M_1,\dots,M_k\) all have the same DP guarantee, then the composition theorem by Kairouz et al.\ \cite{kairouz2017composition} is compatible with \(M_1,\dots,M_k\). If the DP guarantees differ, then the composition theorem by Kairouz et al.\ is not compatible anymore, but the composition theorem by Murtagh et al.\ \cite{complexity2016} is.

In the following we will use \(\compguarantee(M_1,\dots,M_k)\) to denote any DP guarantee for the composition of \(M_1,\dots,M_k\) (in the sense of Def.~\ref{def:composition}). Such a guarantee can be obtained by applying a composition theorem that is compatible with \(M_1,\dots,M_k\) to \(M_1,\dots,M_k\).
We will sometimes take a maximum over different compositions, e.g., \(\max_{i,j} \compguarantee(M_i,M_j)\). Since \(\compguarantee\) returns a tuple, we mean by this fixing a \(\delta\) and then maximizing with respect to \(\varepsilon\), similar to Murtagh et al. \cite{complexity2016}.
In all examples in this subsection the DP guarantee via the classic composition theorem is \(\compguarantee(M_1,\dots,M_k)\) in both the bounded DP and the unbounded DP case.

Throughout the examples we will assume \(|\x_i^0| < |\x_i^1|\) for all \(i\), i.e., database 0 does not contain the target record and database 1 does contain the target record.

\xhdr{Example 1}
In the introduction (\Secref{sec:introduction}) we had the example of \(k\) hospital databases that each get queried once. Let \(M_i\) be the mechanism invoked in iteration \(i\). An individual can contribute to at most \(365\) databases. Therefore, \(\bb\in\tilde{B} = \{\cc\in\{0,1\}^k\mid \sum_{i=i}^k c_i \leq 365\}\). We thus only need to give a guarantee against adversaries with hypotheses that are distributions over \(\tilde{B}\). In the unbounded DP case the adversary wants to determine whether an individual contributed data or not, and hence one of the hypotheses (the one that the individual did not contribute data) will have all of its probability mass on the zero\hyp vector. Let \(p_0^0\) denote this probability distribution, i.e., \(p_0^0(\vec{0}) = 1\), and let \(\mathcal{D}(\tilde{B})\) be the set of all probability distribution over \(\tilde{B}\) as before. Then \(\mathcal{P} = \{p_0^0\}\times \mathcal{D}(\tilde{B})\). As we have seen in \Secref{sec:relationship}, the worst case occurs for non\hyp composite hypotheses. Since \(\bb\) can have at most \(365\) one\hyp entries, we thus get an HDP guarantee of
\begin{equation*}
    \max_{i_1<\dots,<i_{365}} \compguarantee(M_{i_1},\dots,M_{i_{365}}).
\end{equation*}
For bounded DP, i.e., the case where the adversary knows that the target individual has contributed to databases but not to which ones, we do not have the restriction \(p_0(\vec{0})=1\). Thus, \(\mathcal{P} = \mathcal{D}(\tilde{B})\times \mathcal{D}(\tilde{B})\) and the HDP guarantee is
\begin{equation*}
    \max_{i_1<\dots,<i_{730}} \compguarantee(M_{i_1},\dots,M_{i_{730}}),
\end{equation*}
since in the worst case the adversary compares two vectors with \(365\) one\hyp entries each at different positions.

\xhdr{Example 2}
A company has \(k\) different subsidiaries, and for each subsidiary \(i\) a database \(\x_i\) with data from the employees that work in this subsidiary. Since each employee can only work in one of the subsidiaries, an individual's record will either be present in none of the databases or in exactly one database. If each database is queried once, then the set of possible \(\bb\)\hyp vectors is restricted to \(\{\cc\in\{0,1\}^k\mid \norm{\cc}_1 \leq 1\}\). In the unbounded DP case the adversary could thus at worst compare a vector with a single one\hyp entry with a vector with only zeroes, and in the bounded DP case two vectors with a single 1\hyp entry each. This results in an HDP guarantee of
\begin{equation*}
  \max_{i} \compguarantee(M_i)  
\end{equation*}
in the unbounded case (where the composition theorem here just return the DP guarantee of \(M_i\)) and an HDP guarantee of
\begin{equation*}
    \max_{i\neq j} \compguarantee(M_i,M_j)
\end{equation*}
in the bounded case.

\xhdr{Example 3}
We can also treat columns of a database as separate databases. Instead of having restrictions on the databases to which an individual could have contributed, we can then have restrictions on which fields of a column can be non\hyp null. For this example, assume that the database contains usage statistics for a software. The rows correspond to users, the columns to the different features of the software, and the cells contain a number that indicates how often the feature was used by the user. There are two versions of the software: a free version with ads and less functionality, and a paid version without ads and more functionality. There is a set of common features with indices \(1,\dots,k_1\), corresponding to functionality that is present in both versions of the software. Furthermore, there are features \(k_1+1,\dots,k_2\) related to ads (e.g., how often different types of ads have been clicked), and the corresponding database entries are only non\hyp null for users of the free version. Then there are features \(k_2+1,\dots,k_3\) corresponding to functionality that only users of the paid version have access to. Thus, for a given record either the cells corresponding to features \(k_1+1,\dots,k_2\) will be null or the cells corresponding to features \(k_2+1,\dots,k_3\). Assume that in each iteration \(i\) the adversary invokes mechanism \(M_i\) that accesses only the \(i\)\hyp th  column, and ignores null entries in this column (a simple example would be computing a noisy mean). Hence records of free users do not influence the results of queries related to premium functionality, and records of paid users users do not influence the results of queries related to ads. This means that \(\bb\in\tilde{B} = \{\cc^1,\cc^2,\cc^3\}\), where \(\cc^1 = (1,\dots,1,1,\dots,1,0\dots,0)\) (\(k_3-k_2\) zeroes), \(\cc^2 = (1,\dots,1,0\dots,0,1,\dots,1)\) (\(k_2-k_1\) zeroes) and \(\cc^3 = \vec{0}\). Thus, the HDP guarantee in the unbounded case is given as
\begin{align*}
    \max\{&\compguarantee(M_{k_1},\dots,M_{k_3}),\\
    &\compguarantee(M_{1},\dots,M_{k_2})\}
\end{align*}
(\(\cc^1\) vs.\ \(\cc^2\), \(\cc^1\) vs.\ \(\cc^3\), respectively), and in the bounded case as
\begin{align*}
    \max\{&\compguarantee(M_{k_1},\dots,M_{k_3}),\\
    &\compguarantee(M_{1},\dots,M_{k_2}),\\
    &\compguarantee(M_{1},\dots,M_{k_1},M_{k_2},\dots,M_{k_3})\}
\end{align*}
(\(\cc^1\) vs.\ \(\cc^2\), \(\cc^1\) vs.\ \(\cc^3\), \(\cc^2\) vs.\ \(\cc^3\), respectively).

Note that all of these examples work analogously if the adversary queries each database more than once or invokes more than one mechanism per database column.
We want to highlight that the privacy guarantees differ based on whether one works with the neighborhood definition of bounded or unbounded DP.

\xhdr{General case}
Whenever there is a restriction that implies that the maximal number of databases (Examples 1 \& 2) or subsets of the same database (Example 3) that the DP mechanisms are invoked on that an individual can contribute to is smaller than the total number of databases or subsets of the same database, one can improve the privacy bound by using the GCE. In all three examples such restrictions result from the nature of the databases, and are public information: the number of nights in a year (Example 1), the fact that one employee works at exactly one subsidiary (Example 2) and the fact that the free and the paid version of a software have different features (Example 3).
But even in the absence of such publicly known restrictions one can use the GCE for better privacy bounds. In that case, one can treat the restriction as non\hyp private information and estimate it from the data. Take the example of a store chain that wants to compare the performance of its branches by evaluating their sales databases. One might treat the information which products a customer buys as private information, but not the information from how many different branches they buy those products. Then the store chain can compute the maximal number of branches that a single customer has made purchases from, and use this as the upper bound on the number of databases that an individual could have contributed to.
Alternatively, one can also make (reasonable) assumptions that lead to such an upper bound. If instead of shops the branches are restaurants and the data spans one year, then one could make the assumption that one person does not eat at more than three restaurants per day, and could thus only have eaten at at most \(3*365\) different restaurants over the year of the data collection.

\subsubsection{Comparison with parallel composition}
\label{sec:parallel}
In all three examples above, one can also use the parallel composition theorem due to McSherry \cite{mcsherry2009privacy} to improve the privacy guarantee over the naive guarantee via the classic composition experiment. However, the resulting guarantee is at most as good as the one obtained via the GCE, and typically worse. McSherry uses a slightly different definition of DP than the standard one (Def.~\ref{def:dp}). If we were only interested in unbounded DP, we could also use the definition of group privacy \cite{dwork2008theory}, but to also include bounded DP, the definition of McSherry is required:
\begin{definition}[Differential privacy (McSherry) {\cite{mcsherry2009privacy}}]
\label{def:dp_mcsherry}
    A randomized algorithm \(M\) with input domain \(\mathcal{P}(\mathcal{D})\), where \(\mathcal{P}(\mathcal{D})\) is the set of all multisets over some data domain \(\mathcal{D}\), is \(\varepsilon\)\hyp differentially private if, for all \(S\subset\Range(M)\) and for all databases \(\x,\xp\in\mathcal{P}(\mathcal{D})\),
    \begin{equation*}
        \Pr(M(\x)\in S) \leq \exp(\varepsilon \abs{\x\oplus\xp}) \Pr(M(\xp)\in S),
    \end{equation*}
    where \(\oplus\) denotes symmetric difference.
\end{definition}

Parallel composition now states that if we split the database into disjoint parts using a partition of the data domain, then we can apply a DP mechanism to each of the resulting databases, but only pay for one mechanism invocation. In the theorem, the DP definition of McSherry is used.

\begin{theorem}[Parallel composition {\cite{mcsherry2009privacy}}]
    For \(i=1,\dots,k\), let \(M_i\) provide \(\varepsilon\)\hyp DP. Let \(\mathcal{D}_i\), \(i=1,\dots,k\), be disjoint subsets of the data domain \(\mathcal{D}\). The sequence \(M_1(\x\cap \mathcal{D}_1),\dots,M_k(\x\cap \mathcal{D}_2)\) provides \(\varepsilon\)\hyp DP.
\end{theorem}

Note that parallel composition only allows for mechanisms with the same \(\varepsilon\)\hyp DP guarantee. We can extend the theorem to mechanisms \(M_i\) with guarantees \(\varepsilon_i\), where not necessarily \(\varepsilon_i = \varepsilon_j\) for all \(i\neq j\). Then we get a guarantee of \(\max_i \varepsilon_i\) for the sequence.

Parallel composition starts from a single database that is split in a particular way, whereas our examples start with separate databases. To apply parallel composition to our examples, we first need to assemble the separate databases into one database. Assume that there are \(k\) databases \(\x_1,\dots,\x_k\) with data domains \(\tilde{\mathcal{D}}_1,\dots,\tilde{\mathcal{D}}_k\). To each database we add a column that contains the database index, and merge the databases to obtain a single database \(\x\). We then split this database according to the partition \(\mathcal{D}_i = (\bigcup_j \tilde{\mathcal{D}}_j)\times \{i\}\). The mechanism \(M_i\) that we invoke in iteration \(i\) stays the same as before and so does its DP\hyp guarantee. However, we now need to compute the maximal symmetric difference between two neighboring databases \(\x\) and \(\xp\). In the unbounded case of Example 1 (for the other examples it works analogously), \(365\) of the databases \(\x_i\) that make up \(\x\) or \(\xp\) may differ between \(\x\) and \(\xp\). This means that the symmetric difference between \(\x\) and \(\xp\) is \(365\) in the worst case. Applying Def.~\ref{def:dp_mcsherry} then gives us a DP guarantee according to Def.~\ref{def:dp} of \(365\max_i \varepsilon_i\). In the unbounded case, \(730\) of the databases \(\x_i\) may differ between \(\x\) and \(\xp\), and parallel composition hence gives us a guarantee of \(730\max_i \varepsilon_i\). These guarantees are worse than the guarantees via the GCE in three ways: first, they take the maximum over the DP guarantees of all mechanisms as the guarantee for all mechanisms, whereas via the GCE we only need to take the worst \(365\) (or 730) guarantees. Further, parallel composition simply adds up the DP guarantees, wheres the GCE allows us to apply any (advanced) composition theorem, typically resulting in much better bounds. Finally, parallel composition only applies to \(\varepsilon\)\hyp DP mechanisms, but not to \((\varepsilon,\delta)\)\hyp DP mechanisms. This last point, however, has been addressed by extending parallel composition to the even more general setting of \(f\)\hyp DP \cite{smith2022making}.

\subsection{Understanding the Privacy Implications of Limiting the Adversary's Prior Knowledge}
\label{sec:adversaries}
\xhdr{Motivation}
In some cases, there might only be certain hypotheses that we do not want the adversary to be able to distinguish between. Think of our earlier example of an adversary that we want to prevent from learning whether a target record was present in at least one of the databases. Or it might be desirable to give a worst case guarantee with respect to the (from a privacy perspective) worst pair of hypotheses or worst piece of information that the adversary might want to learn, and give better guarantees if the adversary wants to learn a non\hyp worst case piece of information. This allows for more transparency and a better understanding of what exactly a sequence of mechanisms leaks. It also gives new insights into composition theorems, because they can now be compared not only for worst case adversaries.

In the remainder of this section, we first describe how to compute an HDP guarantee with respect to a general set \(\mathcal{P}\). Then we apply this technique to an example with a uniform hypothesis. Finally, we analyze the example in a different way by showing a connection between adversaries with a uniform prior and subsampling.

\xhdr{Computing an HDP guarantee}
For each tuple of distributions \((p_0,p_1)\in\mathcal{P}\), we would like to find a tuple \((\varepsilon(p_0,p_1),\delta(p_0,p_1))\) such that \Eqref{eq:hdp0} and \Eqref{eq:hdp1} in Def.~\ref{def:hdp} hold for these \((p_0,p_1)\) and \((\varepsilon(p_0,p_1),\delta(p_0,p_1))\). To give a guarantee over the entire set \(\mathcal{P}\), we simply take the maximum over all \(\varepsilon(p_0,p_1)\) and \(\delta(p_0,p_1)\) with \((p_0,p_1)\in\mathcal{P}\).

How do we find valid \((\varepsilon(p_0,p_1),\delta(p_0,p_1))\) for a given pair of distributions \((p_0,p_1)\)? We use the notation from the proof of \Thmref{thm:equivalence}: \(p_0\) and \(p_1\) are represented as sets \(T_0 = \{(\bb_0, p_0(\bb_0))\mid \bb_0\in\supp(p_0)\}\) and \(T_1 = \{(\bb_1, p_1(\bb_1))\mid \bb_1\in\supp(p_1)\}\) of tuples of binary vectors and weights. By running \Algref{alg:refine} in \Appref{app:refinement}, we can refine the tuples in \(T_0\) and \(T_1\) to tuples \(U_0\) and \(U_1\) such that there is a perfect matching of the tuples in \(U_0\) and \(U_1\), where the matched tuples have the same weight. More precisely, \(U_0 = \{(\bb_0^1,w_0^1),\dots,(\bb_0^n,w_0^n)\}\) and \(U_1 = \{(\bb_1^1,w_1^1),\dots,(\bb_1^n,w_1^n)\}\) such that \(w_0^i=w_1^i\) for all \(i\). For each pair \(\bb_0^i, \bb_1^i\), we now need to find \((\varepsilon(\bb_0^i, \bb_1^i),\delta(\bb_0^i, \bb_1^i))\) such that
\begin{equation*}
    \Pr(V^{\bb_0^i}\in S) \leq e^{\varepsilon(\bb_0^i, \bb_1^i)} \Pr(V^{\bb_1^i}\in S) + \delta(\bb_0^i, \bb_1^i)
\end{equation*}
and
\begin{equation*}
    \Pr(V^{\bb_1^i}\in S) \leq e^{\varepsilon(\bb_0^i, \bb_1^i)} \Pr(V^{\bb_0^i}\in S) + \delta(\bb_0^i, \bb_1^i)
\end{equation*}
for all sets of views \(S\). With this we can then compute
\begin{align*}
    \delta(p_0,p_1) &= \sum_{i=1}^n w_0^i \delta(\bb_0^i, \bb_1^i),\\
    \varepsilon(p_0,p_1) &= \ln\left[\sum_{i=1}^n w_0^i e^{\varepsilon(\bb_0^i, \bb_1^i)}\right].
\end{align*}

So how to find valid \((\varepsilon(\bb_0^i, \bb_1^i),\delta(\bb_0^i, \bb_1^i))\)? \Thmref{thm:equivalence} implies that we could simply take a standard composition theorem and apply it to the given sequence of mechanisms. However, for \(\bb_0^i, \bb_1^i\) with \(\bb_0^i \oplus \bb_1^i \neq \vec{1}\), this privacy guarantee will be suboptimal. Let \(j(1),\dots,j(k')\) be the indices in which \(\bb_0^i\) and \(\bb_1^i\) coincide. Then the values of the mechanisms in iterations \(j(1),\dots,j(k')\) do not help the adversary in deciding between \(\bb_0^i\) and \(\bb_1^i\), because
\begin{align*}
    &\Pr(V^{\bb_0^i}\in S\mid Y_{j(1)},\dots Y_{j(k')})\\
    &= \Pr(V^{\bb_1^i}\in S\mid Y_{j(1)},\dots Y_{j(k')}).
\end{align*}
Hence we only need to compose over the \(k-k'\) iterations in which \(\bb_0^i\) and \(\bb_1^i\) do not coincide. This can be done by a standard composition theorem.

\xhdr{Example}
Let us analyze the example where we want to prevent the adversary from learning whether a specific individual contributed to at least one of the databases or to none of them. Assume w.l.o.g.\ that \(\x^0_i\) is always the database that does not contain the corresponding record, \(\x^1_i\) is the database that does contain the record, and assume further that all databases are different. If we assume that the adversary has no further knowledge about the membership of the target record, then the hypothesis \(p_0\) that the individual contributed to none of the databases would be given by \(p_0(\vec{0}) = 1\), and the hypothesis \(p_1\) that the individual contributed to at least one database would be the uniform distribution over all \(\bb\in B\setminus \{\vec{0}\}\) (uninformative prior). After tuple refinement, \(U_0=\{(\vec{0},1/(2^k - 1))\}^{2^k-1}\), \(U_1=\{(\bb,1/(2^k - 1))\mid \bb\in B\setminus \{\vec{0}\}\}\). Assume that \(k\) is small and hence simple composition is used as the composition theorem. Assume further that each of the mechanisms is \((\varepsilon,\delta)\)\hyp differentially private. Then for a vector \(\bb\) with \(k'\) one entries, \((\varepsilon(\bb, \vec{0}),\delta(\bb, \vec{0})) = (k' \varepsilon, k' \delta)\). We hence get
\begin{align*}
    \delta(p_0, p_1) &= \frac{1}{2^k - 1} \sum_{j=1}^k \binom{k}{j} j \delta\\
    &= \frac{1}{2^k - 1} k 2^{k-1} \delta
\end{align*}
and
\begin{align}
    \varepsilon(p_0, p_1) &= \ln\left[\frac{1}{2^k - 1} \sum_{j=1}^k \binom{k}{j} e^{j\varepsilon}\right]\label{eq:ineq_delta_p_1}\\
    &= \ln\left[\frac{1}{2^k - 1}\left(\left(1 + e^{\varepsilon}\right)^k-1\right)\right]\label{eq:ineq_delta_p_2},
\end{align}
where in the step from \Eqref{eq:ineq_delta_p_1} to \Eqref{eq:ineq_delta_p_2} we apply the binomial formula with \(x=e^{\varepsilon}\) and \(y=1\), and subtract the missing \(0\)\hyp th summand.
We notice that the improvement over the worst case guarantee of \((k\varepsilon,k\delta)\) is particularly visible for \(\delta\) with a factor of almost \(1/2\).
Here we used simple composition, which will be suboptimal for larger values of \(k\). We now show a more general derivation that allows for using any composition theorem that is compatible with \(M_1,\dots,M_k\).

\xhdr{General derivation for the example using subsampling}
We show a connection between the adversary with the uniform hypothesis from the previous example and subsampling of databases. This allows us to analyze the privacy guarantee against such an adversary using a privacy amplification theorem for subsampling \cite[Thm.~8]{balle2018privacy}.
Assume w.l.o.g.\ that \(\x_i^1 = \x_i^0 \cup x_i\) for some record \(x_i\) for each iteration \(i\). Further assume w.l.o.g.\ that \(\x_i^0 = \emptyset\) (and thus \(\x_i^1 = \{x_i\}\)) for all \(i\). We can assume this w.l.o.g.\ due to the following: Instead of requiring the adversary to return two databases \(\x_i^0,\ \x_i^1\), parameters \(w_i\) and a mechanism \(M_i\) and then have the data curator return \(M_i(\x_i^{\bb_i},w_i)\), we could require the adversary to return only a single record \(x_i\), parameters \(w_i\) and a mechanism \(\tilde{M_i}\). The adversary would define \(\tilde{M_i}\) via \(\tilde{M_i,w}(V) = M_i(\x_i^0\cup V,w)\) for any set of records \(V\) and parameters \(w\), and the data curator would return \(\tilde{M_i}(\emptyset,w_i)\) if \(\bb_i = 0\), and \(\tilde{M_i}(\{x_i\},w_i)\) if \(\bb_i = 0\). This is equivalent to the original interaction between adversary and data curator. Note that this is only a construction within the theoretical framework of composition experiments, which --- even in its original form --- is just a tool to analyze privacy mechanisms; in reality, neither does the adversary give two databases to the data curator, nor do they give the data curator a mechanism with a database baked into it.

In the HDP example that we are analyzing, we want to find \(\varepsilon,\delta\) such that
\begin{align*}
    \mathbb{E}_{\bb\sim U(B\setminus\{\vec{0}\})}[\Pr(V^{\bb}\in S)] &\leq e^{\varepsilon}\Pr(V^{\vec{0}}\in S) + \delta,\\
    \Pr(V^{\vec{0}}\in S) &\leq e^{\varepsilon} \mathbb{E}_{\bb\sim U(B\setminus\{\vec{0}\})}[\Pr(V^{\bb}\in S)] + \delta
\end{align*}
for any output set \(S\), where \(U(B\setminus\{\vec{0}\})\) denotes the uniform distribution over the set of binary vectors with at least one 1 entry. Now consider the following modified data release process: In the setting of the classic composition theorem, in each iteration, every database record is first sampled independently with probability \(1/2\), and the mechanism is then only invoked on the resulting sample of records. Let \(\mathcal{S}_{1/2}(M)\) be the mechanism that is obtained when first sampling each record from the database independently with probability \(1/2\) and then applying \(M\) to the resulting sample. Let then \((\tilde{\varepsilon},\tilde{\delta}) = \compguarantee(\mathcal{S}_{1/2}(M_1),\dots,\mathcal{S}_{1/2}(M_k))\) denote a guarantee for this composition with subsampling. Such a guarantee can be obtained by applying the corresponding amplification by subsampling theorem \cite[Thm.~8]{balle2018privacy} to each mechanism \(M_i\), and then applying a composition theorem that is compatible with \((\mathcal{S}_{1/2}(M_1),\dots,\mathcal{S}_{1/2}(M_k))\). Since in our case there is only a single record \(x_i\), subsampling is equivalent to choosing either \(\x_i^0\) or \(\x_i^1\), each with probability \(1/2\). Thus, we have
\begin{align*}
    &\Pr((r,M_1(\x_1),\dots,M_k(\x_k))\in S\\
    &\hphantom{\Pr((} \mid \x_1\in U(\x_1^0,\x_1^1),\dots,\x_k\in U(\x_k^0,\x_k^1))\\
    &\leq e^{\tilde{\varepsilon}} \Pr((r,M_1(\x_1^0),\dots,M_k(\x_k^0))\in S) + \tilde{\delta}
\end{align*}
for any output set \(S\), and the inverse inequality with the two probabilities exchanged. With this, we can show the following HDP guarantee for our example of an adversary that compares the \(\vec{0}\) hypothesis with the hypothesis that is the uniform distribution over all non\hyp zero binary vectors (proof in \Appref{app:proof_subsampling}):
\begin{theorem}
\label{thm:subsampling}
A sequence of mechanisms \(M_1,\dots,M_k\) fulfills \((\varepsilon,\delta)\)\hyp HDP with respect to \(\mathcal{P}=\{([p_0(\vec{0})=1, p_0(\cc)=0\text{ for all }\cc\neq\vec{0}], U(B\setminus \{\vec{0}\}))\} \) for
\begin{equation*}
    \varepsilon = \ln\left[\frac{1}{2^k - 1} \sum_{i=1}^k 2^{k-i}e^{\hat{\varepsilon}_i}\right],\quad
    \delta = \frac{1}{2^k - 1} \sum_{i=1}^k 2^{k-i}\hat{\delta}_i,
\end{equation*}
where
\begin{equation*}
    (\hat{\varepsilon}_i,\hat{\delta}_i) = \compguarantee(M_i,\mathcal{S}_{1/2}(M_{i+1}),\dots,\mathcal{S}_{1/2}(M_k)).
\end{equation*}
\end{theorem}

Using simple composition for composing the one non\hyp subsampled mechanism in each summand with the subsampled mechanisms, we get the following corollary:
\begin{corollary}
A sequence of mechanisms \(M_1,\dots,M_k\) fulfills \((\varepsilon,\delta)\)\hyp HDP with respect to \(\mathcal{P}=\{([p_0(\vec{0})=1, p_0(\cc)=0\text{ for all }\cc\neq\vec{0}], U(B\setminus \{\vec{0}\})) \} \) for
\begingroup
\allowdisplaybreaks 
\begin{align*}
    \varepsilon &= \ln\left[\frac{1}{2^k - 1} \sum_{i=1}^k 2^{k-i}e^{\varepsilon_i + \hat{\varepsilon}_i}\right],\\
    \delta &= \frac{1}{2^k - 1} \sum_{i=1}^k 2^{k-i}(\delta_i + \hat{\delta}_i),
\end{align*}
\endgroup
where
\begin{equation*}
    (\hat{\varepsilon}_i,\hat{\delta}_i) = \compguarantee(\mathcal{S}_{1/2}(M_{i+1}),\dots,\mathcal{S}_{1/2}(M_k)).
\end{equation*}
\end{corollary}

By composing the subsampled mechanisms via simple composition and using \cite[Thm.~8]{balle2018privacy} for the amplification via subsampling, we can recover the results from the original derivation that we did earlier in this example:
\begin{corollary}
Let \(M_1,\dots,M_k\) be a sequence of mechanisms with \(\varepsilon_i = \varepsilon'\) and \(\delta_i = \delta'\) for all \(i\) and fixed \(\varepsilon',\delta'\). Then \(M_1,\dots,M_k\) fulfills \((\varepsilon,\delta)\)\hyp HDP with respect to \(\mathcal{P}=\{([p_0(\vec{0})=1, p_0(\cc)=0\text{ for all }\cc\neq\vec{0}], U(B\setminus \vec{0})) \} \) for
\begin{equation*}
    \varepsilon = \ln\left[\frac{\left(e^{\varepsilon'}+1\right)^k-1}{2^k - 1}\right],\quad
    \delta = \frac{2^{k-1}}{2^k - 1} k \delta'.
\end{equation*}
\end{corollary}
\begin{proof}
Plugging in \(\delta/2\) from the amplification theorem and then composing \(k-i\) times, we get
\begin{align}
    \nonumber
    \delta &= \frac{1}{2^k - 1}\sum_{i=1}^k 2^{k-i} \left(1+\frac{k-i}{2}\right)\delta'\\
    &= \frac{\delta'}{2^k - 1} \left( \sum_{j=0}^{k-1} 2^{j} + \sum_{j=0}^{k-1} j 2^{j-1}  \label{eq:c6d1} \right) \\
    &= \frac{2^{k-1}}{2^k - 1} k \delta' \ , \label{eq:c6d2}
\end{align}
where to go from~\Eqref{eq:c6d1} and~\Eqref{eq:c6d2}, we use the fact that the first sum is a partial geometric series $\sum_{i=0}^n x^i = \frac{1 - x^{n+1}}{1 - x}$ for $x \neq 1$. The second sum can be simplified by taking the derivative of both side of the previous expression with respect to $x$.

Plugging in \(\ln\left(1+\frac{1}{2}(e^{\varepsilon'}-1)\right)\) from the amplification theorem and then composing \(k-i\) times, we get
\begin{equation*}
    \varepsilon = \ln\left[\frac{1}{2^k - 1}\sum_{i=1}^k 2^{k-i} \exp\left[\varepsilon'+(k-i)\ln\left(1+\frac{1}{2}(e^{\varepsilon'}-1)\right)\right]\right]
\end{equation*}
Observe that: 
\begin{align*}
& 2^{k-i} \exp{[\varepsilon'+(k-i)\ln(1+\frac{1}{2}(e^{\varepsilon'}-1) ] } \\
&= e^{\varepsilon'} \cdot 2^{k-i} \exp{[(k-i)\ln(1+\frac{1}{2}(e^{\varepsilon'}-1)]} \\
&= e^{\varepsilon'} \cdot (1+e^{\varepsilon'})^{k-i} .
\end{align*}

So that:
\begin{align*}
    &\sum_{i=1}^k 2^{k-i} \exp\left[\varepsilon'+(k-i)\ln\left(1+\frac{1}{2}(e^{\varepsilon'}-1)\right)\right] \\
    &= e^{\varepsilon'} \sum_{i=1}^k (1+e^{\varepsilon'})^{k-i} \\
    &= (e^{\varepsilon'} + 1)^{k} - 1 \ ,
\end{align*}
which yields the result.
\end{proof}

\section{Conclusions}
\label{sec:conclusions}
Composition is an essential property for a privacy metric --- without it, one cannot reason about the privacy over multiple data releases. In this paper, we point out that the classic composition experiment (used by prior research to prove composition results) does not allow for maximum flexibility in terms of the adversary's choice of databases with respect to the underlying hypotheses.

To address this limitation, we propose a generalized composition experiment (GCE) that enables us to reason about differential privacy guarantees in applications that simply cannot be modeled with the classic composition experiment. We prove that the worst case guarantee of the classic composition experiment coincides with the worst case guarantee for the GCE. This is significant because it implies that despite its limitation, the classic composition experiment does not underestimate the privacy risk. 

However, the added flexibility of the GCE is significant. We analyze two applications and show that the the GCE can: (1) yield improved privacy guarantees when there are constraints on the individual's contribution to the different databases, and (2) analyze the privacy loss with respect to adversaries with varying amount of prior knowledge, thereby allowing us to reason about relaxations of the very strong background knowledge assumption of differential privacy. In this context, we uncover an intriguing connection between adversarial background knowledge and the boosting effects of database subsampling on the differential privacy guarantees.


\bibliographystyle{IEEEtran}
\bibliography{references}

\appendices

\section{Refinement Algorithm}
\label{app:refinement}
Both in the proof of \Thmref{thm:equivalence} and in \Secref{sec:adversaries} we are given two distribution \(p_0\) and \(p_1\) over \(B\) and want to apply inequalities that hold with respect to individual points in their supports. \(p_0\) and \(p_1\) can be expressed in terms of the sets of tuples \(T_0 = \{(\bb_0, p_0(\bb_0))\mid \bb_0\in\supp(p_0)\}\) and \(T_1 = \{(\bb_1, p_1(\bb_1))\mid \bb_1\in\supp(p_1)\}\). In order to be able to apply the inequalities, we need to be able to match tuples \((\bb_0,w_0)\in T_0\) and \((\bb_1,w_1)\in T_1\), where \(w_0=w_1\). However, this will not always be possible, since the distributions over probability masses need not be the same for \(p_0\) and \(p_1\). We can solve this problem by refining tuples via \Algref{alg:refine}. It takes as input the sets \(T_0\) and \(T_1\), and in each step refines a tuple \((\bb,w)\) from one of the sets into two tuples \((\bb,w^1),\ (\bb,w^2)\) with \(w = w^1+w^2\). It returns two sets \(U_0,\ U_1\) of such refined tuples for which a matching can be found.

\begin{algorithm}[H]
\caption{\(\refinetuples(T_0,T_1)\)}
\label{alg:refine}
\begin{algorithmic}[1]
\State{\textbf{Input:} \(T_0, T_1\)}
\State{Initialize \(U_0, U_1 \defeq \emptyset\)}
\While{\(T_0\neq\emptyset\)}
    \State{Take and remove some tuple \((\bb_0,w_0)\) from the set \(T_0\) and some tuple \((\bb_1,w_1)\) from the set \(T_1\)}
    \If{\(w_0\leq w_1\)}
        \State{Add \((\bb_0,w_0)\) to \(U_0\) and \((\bb_1,w_0)\) to \(U_1\)}
        \State{Add \((\bb_1,w_1-w_0)\) to \(T_1\)}
    \Else
        \State{Add \((\bb_0,w_1)\) to \(U_0\) and \((\bb_1,w_1)\) to \(U_1\)}
        \State{Add \((\bb_0,w_0-w_1)\) to \(T_0\)}
    \EndIf
\EndWhile
\Return \(U_0,U_1\)
\end{algorithmic}
\end{algorithm}

\section{Proof of \Thmref{thm:equivalence}}
\label{app:proof_equivalence}
\begin{proof}[Proof of \Thmref{thm:equivalence}]
\textit{HDP implies DP:}
Classic composition assumes that there is only a single bit \(b\), which is equivalent to assuming \(\bb=\vec{0}\) or \(\bb=\vec{1}\) in the generalized composition experiment. Assume that \(\mathcal{P} = \mathcal{D}(B)\times \mathcal{D}(B)\) and that the sequence of mechanisms fulfills \((\varepsilon,\delta)\)\hyp HDP. \(\mathcal{P}\) hence contains the pair of distributions \([p_0(\vec{0}) = 1,\ p_0(\bb) = 0\text{ for all }\bb\neq \vec{0}]\) and \([p_1(\vec{1}) = 1,\ p_1(\bb) = 0\text{ for all }\bb\neq \vec{1}]\). We therefore have
\begin{align*}
    \Pr(V^{\vec{0}}\in S) &= \mathbb{E}_{\bb\sim p_0}[\Pr(V^{\bb}\in S)]\\
    &\leq e^{\varepsilon} \mathbb{E}_{\bb\sim p_1}[\Pr(V^{\bb}\in S)] + \delta\\
    &= e^{\varepsilon} \Pr(V^{\vec{1}}\in S) + \delta,
\end{align*}
and the same with 0 and 1 exchanged.
The sequence of mechanisms hence also fulfills \((\varepsilon,\delta)\)\hyp DP.

\textit{DP implies HDP:}
Fix an adversary \(\mathcal{A}\), i.e., a (randomized) algorithm that chooses mechanisms and databases.

To simplify the proof, assume w.l.o.g.\ that the order of the databases \(\x_i^0\) and \(\x_i^1\) in each step is fixed according to some deterministic criterion, e.g., the database size. Instead of choosing the order of the databases within each pair \(\x_i^0,\x_i^1\) of databases and comparing the likelihood of the 0- and of the 1\hyp bit in the classic composition experiment, the adversary may now choose an arbitrary \(\bb_0\in B\) and compare the likelihood of \(\bb_0\) and the vector that results from flipping all bits in \(\bb_0\), which we denote by \(\bb_0'\). In general we denote the vector resulting from flipping all bits in a vector \(\bb\) by \(\bb'\).

First consider the case where \(p_0(\bb_0) = 1\) for some \(\bb_0\) and \(p_1(\bb_1) = 1\) for some \(\bb_1\), i.e., the two distributions are deterministic. Fix the set \(S\). Assume w.l.o.g.\ that \(\Pr(V^{\bb_0}\in S) \leq \Pr(V^{\bb_1}\in S)\). This immediately implies \(\Pr(V^{\bb_0}\in S) \leq e^{\varepsilon} \Pr(V^{\bb_1}\in S) + \delta\). Our goal is to flip bits in both \(\bb_0\) and \(\bb_1\) until we arrive at vectors \(\cc_0\) and \(\cc_1\) with \(\cc_1 = \cc_0'\), which we will use to show that also the reverse inequality \(\Pr(V^{\bb_1}\in S) \leq e^{\varepsilon} \Pr(V^{\bb_0}\in S) + \delta\) holds. Consider the set \(I_{\text{same}}\subset [k]\) of all indexes corresponding to bits that have the same value in \(\bb_0\) and \(\bb_1\). For each such index \(i\in I_{\text{same}}\), one of the following three cases will occur:
\begin{enumerate}
    \item Flipping \(\bb_{0,i}\) decreases \(\Pr(V^{\bb_0}\in S)\).\label{enum:flip1}
    \item Flipping \(\bb_{1,i}\) increases \(\Pr(V^{\bb_1}\in S)\).\label{enum:flip2}
    \item Flipping \(\bb_{0,i}\) leaves \(\Pr(V^{\bb_0}\in S)\) unchanged.\label{enum:flip3}
\end{enumerate}
This holds because of the following. Since the internal randomness of the mechanisms is independent, the outputs of the mechanisms are independent (only the choice of mechanisms depends on the history). We can hence treat the hypotheses in each iteration independently. The output \(y_i\) in iteration \(i\) will either (a) support hypothesis \(\bb_{0,i}\) stronger than \(\bb_{0,i}'\); (b) support hypothesis \(\bb_{0,i}' = \bb_{1,i}'\) stronger than hypothesis \(\bb_{0,i} = \bb_{1,i}\); or (c) support both hypothesis equally. (a), (b) and (c) correspond to the three cases \ref{enum:flip1}, \ref{enum:flip2} and \ref{enum:flip3} in the enumeration above. Except for one specialty for (a), namely when \(\Pr(V^{\bb_0}\in S) = 0\), because then the probability cannot be decreased further. Then (a) corresponds to case \ref{enum:flip3}.

For each index \(i\in I_{\text{same}}\), modify \(\bb_0\) and \(\bb_1\) according to the rules: If \ref{enum:flip1} holds, flip \(\bb_{0,i}\); if \ref{enum:flip1} does not hold but \ref{enum:flip2} holds, flip \(\bb_{1,i}\); if neither \ref{enum:flip1} nor \ref{enum:flip2} holds, flip \(\bb_{0,i}\). This results in a sequence of pairs of vectors \((\bb_0^0, \bb_1^0), (\bb_0^1, \bb_1^1),(\bb_0^2, \bb_1^2),\dots,(\cc_0, \cc_1)\), where \(\cc_1 = \cc_0'\). Because we assume that the inequality in Def.~\ref{def:composition} holds, \(\Pr(V^{\cc_1}\in S) \leq e^{\varepsilon} \Pr(V^{\cc_0}\in S) + \delta\). Due to how we flipped bits, we have \(\Pr(V^{\cc_0}\in S)\leq\dots\leq \Pr(V^{\bb_0^1}\in S)\leq\Pr(V^{\bb_0^0}\in S)\) and \(\Pr(V^{\bb_1^0}\in S)\leq\Pr(V^{\bb_1^1}\in S)\leq\dots\leq\Pr(V^{\cc_1}\in S)\). Thus \(\Pr(V^{\cc_1}\in S) \leq e^{\varepsilon} \Pr(V^{\cc_0}\in S) + \delta\) implies \(\Pr(V^{\bb_1}\in S) \leq e^{\varepsilon} \Pr(V^{\bb_0}\in S) + \delta\). Hence Def.~\ref{def:composition} implies Def.~\ref{def:hdp} for the special case of deterministic \(p_0\) and \(p_1\).

But this also implies Def.~\ref{def:hdp} for the general case where \(p_0\) and \(p_1\) are distributions over multiple vectors. From what we have shown so far, we have that, for all \(\bb_0\in \supp(p_0)\) and all \(\bb_1\in \supp(p_1)\), \(\Pr(V^{\bb_0}\in S) \leq e^{\varepsilon} \Pr(V^{\bb_1}\in S) + \delta\).
Consider the sets of tuples \(T_0 = \{(\bb_0, p_0(\bb_0))\mid \bb_0\in\supp(p_0)\}\) and \(T_1 = \{(\bb_1, p_1(\bb_1))\mid \bb_1\in\supp(p_1)\}\), and \Algref{alg:refine} from \Appref{app:refinement}.

We show the statement for the general case by showing an inequality for \(U_1\) and \(U_1\) by induction over the iteration number of the while loop, and then relating this inequality to \(T_0\) and \(T_1\). We index the state of \(U_0\) and \(U_1\) after the \(i\)\hyp th iteration by \(U_0^i\) and \(U_1^i\).

\xhdr{Base case} For \(U_0^0 = U_1^0 = \emptyset\) we have \(\sum_{(\bb_0,w_0)\in U_0^0} w_0 \Pr(V^{\bb_0}\in S) = 0 \leq 0 = e^{\varepsilon} \sum_{(\bb_1,w_1)\in U_1^0} w_1 \Pr(V^{\bb_1}\in S) + w_1 \delta\).

\xhdr{Induction step} Assume that
\begin{align*}
    &\sum_{(\bb_0,w_0)\in U_0^i} w_0 \Pr(V^{\bb_0}\in S)\\
    &\leq e^{\varepsilon} \sum_{(\bb_1,w_1)\in U_1^i} w_1 \Pr(V^{\bb_1}\in S) + \delta \sum_{(\bb_1,w_1)\in U_1^i} w_1.
\end{align*}
We look at the next iteration of the while loop and denote the tuples chosen in the \(i+1\)\hyp th iteration with the superscript \(i+1\). Assume w.l.o.g.\ that \(w_0^{i+1}\leq w_1^{i+1}\). Because of what we have shown so far, we know that \(\Pr(V^{\bb_0^{i+1}}\in S) \leq e^{\varepsilon} \Pr(V^{\bb_1^{i+1}}\in S) + \delta\) and hence
\begin{equation*}
    w_0^{i+1} \Pr(V^{\bb_0^{i+1}}\in S) \leq e^{\varepsilon} w_0^{i+1} \Pr(V^{\bb_1^{i+1}}\in S) + w_0^{i+1} \delta.
\end{equation*}
Thus
\begin{align*}
    &\sum_{(\bb_0,w_0)\in U_0^{i+1}} w_0 \Pr(V^{\bb_0}\in S)\\
    &= \sum_{(\bb_0,w_0)\in U_0^i} w_0 \Pr(V^{\bb_0}\in S) + w_0^{i+1} \Pr(V^{\bb_0^{i+1}}\in S)\\
    &\leq e^{\varepsilon} \sum_{(\bb_1,w_1)\in U_1^i} w_1 \Pr(V^{\bb_1}\in S)\\
    &\hphantom{\leq} + \delta \sum_{(\bb_1,w_1)\in U_1^i} w_1 + w_0^{i+1} \Pr(V^{\bb_0^{i+1}}\in S)\\
    &\leq e^{\varepsilon} \sum_{(\bb_1,w_1)\in U_1^i} w_1 \Pr(V^{\bb_1}\in S) + \delta \sum_{(\bb_1,w_1)\in U_1^i} w_1\\
    &\hphantom{\leq} + e^{\varepsilon} w_0^{i+1} \Pr(V^{\bb_1^{i+1}}\in S) + w_0^{i+1} \delta\\
    &= e^{\varepsilon} \sum_{(\bb_1,w_1)\in U_1^{i+1}} w_1 \Pr(V^{\bb_1}\in S) + \delta \sum_{(\bb_1,w_1)\in U_1^{i+1}} w_1.
\end{align*}
This shows the induction step.

Consider \(U_0\) and \(U_1\) after the last iteration. They now contain all tuples from \(T_0\) and \(T_1\), just some of them split up. Further, \(\sum_{(\bb_1,w_1)\in U_1} w_1 = 1\). Hence
\begin{equation*}
    \sum_{(\bb_0,w_0)\in T_0} w_0 \Pr(V^{\bb_0}\in S) \leq \sum_{(\bb_1,w_1)\in T_1} w_1 \Pr(V^{\bb_1}\in S) + \delta
\end{equation*}
and thus, by definition of \(T_0\) and \(T_1\),
\begin{align*}
    &\sum_{\bb_0\in \supp(p_0)} p_0(\bb_0) \Pr(V^{\bb_0}\in S)\\
    &\leq e^{\varepsilon} \sum_{\bb_1\in \supp(p_1)} p_1(\bb_1) \Pr(V^{\bb_1}\in S) + \delta.
\end{align*}
This concludes the proof.
\end{proof}

\section{Proof of \Thmref{thm:subsampling}}
\label{app:proof_subsampling}
\begin{proof}[Proof of \Thmref{thm:subsampling}]
We have
\begin{align*}
    &\Pr((r,M_1(\x_1),\dots,M_k(\x_k))\in S\mid \x_i\in U(\x_i^0,\x_i^1))\\
    &= \Pr((r,M_1(\x_1^{b_1}),\dots,M_k(\x_k^{b_k}))\in S\mid \bb\in U(B))\\
    &= \mathbb{E}_{\bb\sim U(B)}[\Pr(V^{\bb}\in S)]
\end{align*}
and also
\begin{equation*}
    \Pr((r,M_1(\x_1^0),\dots,M_k(\x_k^0))\in S) = \Pr(V^{\vec{0}}\in S),
\end{equation*}
and hence
\begin{equation*}
    \mathbb{E}_{\bb\sim U(B)}[\Pr(V^{\bb}\in S)] \leq e^{\tilde{\varepsilon}} \Pr(V^{\vec{0}}\in S) + \tilde{\delta}
\end{equation*}
and
\begin{equation*}
    \Pr(V^{\vec{0}}\in S) \leq e^{\tilde{\varepsilon}} \mathbb{E}_{\bb\sim U(B)}[\Pr(V^{\bb}\in S)] + \tilde{\delta}.
\end{equation*}
We further have
\begin{equation*}
    \mathbb{E}_{\bb\sim U(B)}[\Pr(V^{\bb}\in S)] = \frac{1}{2^k} \sum_{\bb\in B} \Pr(V^{\bb}\in S)
\end{equation*}
and
\begin{equation*}
    \mathbb{E}_{\bb\sim U(B\setminus\{\vec{0}\})}[\Pr(V^{\bb}\in S)] = \frac{1}{2^k - 1} \sum_{\bb\in B\setminus\{\vec{0}\}} \Pr(V^{\bb}\in S).
\end{equation*}
We will use this together with the following division of \(B\setminus\{\vec{0}\}\):
\begin{align*}
    B\setminus\{\vec{0}\}&=\bigcup_{i=1}^{k} B_i\text{, where}\\
    B_i&=\{\bb\in\{0,1\}^k\mid b_j = 0\text{ for all } j<i,\ b_i=1\}.
\end{align*}
Note that this is a disjoint union. We have \(|B_i| = 2^{k-i}\). Hence,
\begin{align*}
    &\mathbb{E}_{\bb\sim U(B\setminus\{\vec{0}\})}[\Pr(V^{\bb}\in S)]\\
    &= \frac{1}{2^k - 1} \sum_{b\in B\setminus\{\vec{0}\}} \Pr(V^{\bb}\in S)\\
    &= \frac{1}{2^k - 1} \sum_{i=1}^k \sum_{b\in B_i} \Pr(V^{\bb}\in S)\\
    &= \frac{1}{2^k - 1} \sum_{i=1}^k 2^{k-i}\frac{1}{2^{k-i}}\sum_{b\in B_i} \Pr(V^{\bb}\in S)\\
    &= \frac{1}{2^k - 1} \sum_{i=1}^k 2^{k-i}\mathbb{E}_{\bb\sim U(B_i)}[\Pr(V^{\bb}\in S)]\\
    &= \frac{1}{2^k - 1} \sum_{i=1}^k 2^{k-i}\Pr((r,M_1(\x_1^0),\dots,M_{i-1}(\x_{i-1}^0),M_i(\x_i^1),\\
    &\phantom{= \frac{1}{2^k - 1} \sum_{i=1}^k 2^{k-i}\Pr(} M_{i+1}(\x_{i+1})\dots,M_k(\x_k))\in S\\
    &\phantom{= \frac{1}{2^k - 1} \sum_{i=1}^k 2^{k-i}\Pr(} \mid \x_j\in U(\x_j^0,\x_j^1)\text{ for }i+1\leq j\leq k)).
\end{align*}
For HDP, we compare this quantity with \(\Pr((r,M_1(\x_1^0),\dots,M_k(\x_{i-1}^0)\in S)\), and thus get the guarantee from the theorem statement.
\end{proof}

\end{document}